%% file: main.tex
\documentclass[sigconf]{acmart}

\usepackage{graphicx} 
\usepackage{algorithm}
\usepackage[noend]{algpseudocode}
\usepackage{xspace}
\usepackage{subcaption}
\usepackage{multirow}
\usepackage{doi}

\usepackage{cuted}

\setcopyright{none}
\renewcommand\footnotetextcopyrightpermission[1]{}
\title{Efficient Heuristics and Machine Learning Approach for Fault Characterization in Distributed Self-Stabilizing Programs}

\author{Amit Garu, Duong Nguyen}
\email{agaru@uwyo.edu, duong.nguyen@uwyo.edu}
\affiliation{%
  \institution{University of Wyoming}
  \city{Laramie}
  \state{Wyoming}
  \country{USA}
}

\newcommand{\cvf}{\textit{cvf}\xspace}
\newcommand{\Cvf}{\textit{Cvf}\xspace}

\begin{abstract}
\input{abstract}
\end{abstract}

  
\begin{document}

\maketitle


\section{Introduction}

\input{introduction}


\section{System Model and Architecture}
\input{sys-model}

\section{Experimental Setting}
\input{experiment-setting}

\section{Results \& Discussions}
\input{result}

\section{Related Works}
\input{related-work}

\section{Conclusion and Future Works}
\input{conclusion}





\bibliographystyle{splncs04}

\bibliography{amit,paper}

\clearpage

\begin{appendix}
\section*{Appendixes}
\input{appendixes}
\end{appendix}

\end{document}

%% file: abstract.tex
Modern large-scale systems rely on distributed protocols to maximize efficiency while preserving the correctness guarantees of single-process execution. However, designing such protocols is non-trivial: adding resources to a system inherently increases its complexity, which in turn introduces faults that must be addressed during design. One such fault class, arising in systems that use distributed shared memory (e.g., replicated databases), is consistency violating fault (\cvf{}), a fault in which a process accessing shared memory reads stale data previously written by another process. \Cvf{}s are inherent to systems that prioritize availability over strict consistency. Prior work has shown that self-stabilizing programs remain correct under high availability, provided such faults stay within a tolerable bound. Characterizing the behavior of self-stabilizing programs in the presence of \cvf{}s can therefore help system designers build systems that are both correct and efficient. In this paper, we study self-stabilizing programs under \cvf{}s to understand \emph{when} and \emph{where} such faults occur, enabling runtime measures that improve computational efficiency. Since exhaustive state-space exploration does not scale, it suffers from state-space explosion; our study investigates two complementary approaches: a conflict-based state-space exploration technique for monotonically stabilizing ($\Delta$ + 1)-Graph Coloring programs, and a machine-learning-based approach for a Maximal Independent Set program. We propose a conflicting-edges-based state-space exploration algorithm alongside a rigorously validated ML model, evaluated on an out-of-distribution dataset through structural consistency checks against known graph invariants.



%% file: introduction.tex
Distributed computation is central to modern systems: as data and problem sizes grow, serial computation becomes non-scalable, motivating the design of algorithms suited to distributed execution. However, designing a distributed algorithm that is a correct counterpart to its serial version is itself a non-trivial task, since distributed systems are inherently complex, difficult to reason about, and susceptible to faults. This study analyzes one such fault, the Consistency Violation Fault (\cvf{})~\cite{Nguyen2019ICDCN,NguyenGuptaKulkarni23ICDCNAnalyzing}, which arises in replicated systems, such as geo-replicated data stores, where a resource's true (current) value is distributed across locations that may not be instantly synchronized. A read is \emph{consistent} if it returns the resource's current value; it constitutes a \cvf{} when a node instead acts on a stale value read from a replica. This study examines \cvf{} through self-stabilizing graph problems, since self-stabilizing programs are designed to tolerate such transient inconsistency.

The effect of a \cvf{} can vary substantially with its location $where$ and timing $when$ in the system. It is natural to expect that faults at more central nodes are more destructive than faults at peripheral ones, though identifying which nodes are central in this sense is itself non-trivial and depends on the specific self-stabilizing problem under study. Cascading \cvf{}s may also have a disproportionately larger effect than isolated faults, and a \cvf{} occurring near the invariant may be more damaging than one occurring far from it, since a configuration close to legitimacy has comparatively little room to absorb a perturbation without moving further away. This study aims to answer several such research questions: given a self-stabilizing problem and an input graph, which nodes are most susceptible to damaging faults, which transitions warrant the most attention, and how self-stabilizing graph problems can be solved more reliably in distributed systems. Answering these questions, however, requires exhaustively enumerating the full configuration space of the system, which quickly becomes computationally intractable as graph size grows.

This intractability is not unique to \cvf{} analysis: exhaustive state space exploration is a well-known bottleneck in the analysis and verification of concurrent systems generally, since the number of reachable states typically grows exponentially with the number of state variables in the system~\cite{valmari1998state}. For example, a self-stabilizing graph coloring algorithm on a 10-node complete graph must explore on the order of $10^9$ states. This is already too many states to store on a modern system, making an exhaustive search infeasible. This is known as the \emph{State Space Explosion} problem, and much of the research in model checking over the past several decades has focused on developing techniques to mitigate it~\cite{pelanek2009fighting,clarke2012model}. 

In our setting, large-scale graph processing is carried out using a \emph{Distributed Shared Memory} (DSM) abstraction, which gives processes on separate machines the illusion of a single, shared address space~\cite{LiHudak89TOCS}. Since the underlying data is not actually co-located, DSM systems need a memory coherence protocol to keep reads and writes consistent across machines, historically implemented via centralized or distributed manager algorithms~\cite{LiHudak89TOCS}. In practice, this abstraction is often implemented atop a replicated storage layer for low-latency access and fault tolerance, which introduces further consistency challenges of its own.

A replicated system maintains multiple copies of shared data across geographically distributed machines, enabling low-latency access and fault tolerance in the presence of network partitions and node failures. This model underlies much of modern distributed infrastructure, from edge computing and compliance-driven architectures to large-scale cloud services; geo-replicated databases such as Amazon's Dynamo, Cassandra, and Google Spanner are prominent industrial examples~\cite{DeCandiaHastorun07, cassandra2009, spanner2013}. A central challenge in these systems is managing the trade-off between consistency, availability, and latency across replicas~\cite{GilbertLynch02Brewer}, which motivates the range of consistency models, from strong consistency guarantees to weaker, eventually consistent semantics, considered in this work.

As data volumes and problem sizes continue to grow, large-scale graph processing has become a cornerstone of modern distributed systems built on this kind of infrastructure. \emph{Distributed Graph Computation} is widely employed for algorithms such as PageRank and Triangle enumeration~\cite{pandurangan2018distributedcomplexitylargescalegraph}, as well as shortest-path, bipartite matching, and clustering algorithms implemented within vertex-centric frameworks such as Pregel~\cite{MalewiczMatthew10ICMD}. Computation over data local to a node is fast and avoids the risk of acting on outdated information, since the data is directly and immediately available. When computation instead depends on data held at other locations, the possibility of \cvf{} arises, since a node may act on a stale copy of remote data before it is updated. Self-stabilizing algorithms are one prominent approach to designing graph computations that remain correct despite such faults, which motivates our focus on \cvf{} analysis in this setting.



%% file: sys-model.tex
\subsection{Self stabilization}
Given $P$ as a distributed program, $S_p$ as the program's state space, $s_i$ as a program state in $P$, and $\delta_p$ as the program's state transition function. Suppose $I \subseteq S_p$ is a set of the program's invariants (legitimate states where the program's specification is met). Then program $P$ is self-stabilizing with respect to $I$ if it satisfies two properties ~\cite{Dijkstra1974CACM}:
\begin{enumerate}
    \item 
    \textbf{Convergence:} from an arbitrary initial state, every computation reaches a state in $I$ within a finite number of steps:\\
    $\forall s \in S_P: \exists s_0 = s, s_1, s_2, ..., s_k \in S_P \land (s_i, s_{i+1}) \in \delta_P \; \forall i \geq 0 \land s_k \in I$
    
    \item \textbf{Closure:} Any program transition from a state inside the invariant will lead to another state also inside the invariant.
    $\forall s \in I \land (s, s') \in \delta_P \Rightarrow s' \in I$.
\end{enumerate}

\subsection{Static Program Analysis}
\input{spa}

\subsection{Case Study of Distributed Programs}
\subsubsection{$(\Delta+1)$-Graph Coloring (GC)}
For CE-based full analysis, we study a monotonically stabilizing $(\Delta+1)$-Graph Coloring algorithm from \cite{HEDETNIEMI2003251}. Each node's state is the numeric color it currently holds, drawn from $\{1, 2, \ldots, \Delta+1\}$, where $\Delta$ is the maximum degree of $G$. A node need not hold the minimum color among its neighbors; correctness only requires that its color differs from every neighbor's color.

\subsubsection{Maximal Independent Set (MIS)}
For ML-based \cvf{} analysis, we study the MIS algorithm from \cite{TURAU200788}. Each node's state takes one of three values, $\{\mathrm{IN}, \mathrm{OUT}, \mathrm{WAIT}\}$. A configuration of $(MIS, G)$ is legitimate when (1) no two neighboring nodes are both labeled $\mathrm{IN}$ (independence), (2) every node labeled $\mathrm{OUT}$ has at least one neighbor labeled $\mathrm{IN}$ (maximality), and (3) no node is labeled $\mathrm{WAIT}$ (all nodes have settled).

%% file: spa.tex


Static program analysis (SPA) is a technique for reasoning about a program's properties by examining its structure and specification directly, without executing it~\cite{GosainSharma15StaticAnalysis}. Because it forgoes execution, static analysis is typically faster and less computationally expensive than dynamic, execution-based approaches. We adopt this technique primarily for the generalization it affords: conclusions drawn from static analysis extend across graphs of varying size and topology, rather than being tied to a specific instance evaluated at runtime.

\textbf{SPA on self-stabilizing algorithms.} Let $P$ denote a self-stabilizing program and $G$ the graph over which $P$ executes, and let $(P, G)$ denote the resulting system under study. We define the \emph{configuration space} of $(P, G)$ as a forest of trees, in which each edge represents a single program transition: the parent node is the source configuration, and each child node is a configuration obtained by applying one step of $P$ to the parent. The leaf nodes of the forest are exactly those configurations satisfying the \texttt{Invariant} predicate, corresponding to the stabilized (legitimate) configurations of $(P, G)$.

\begin{algorithm}[htbp]
\small
\caption{\small Static Program Analysis (computation of Ranks) of any self-stabilizing program using Full Analysis.}
\label{alg:spa-fa-analysis}
\begin{algorithmic}[1]
\Procedure{FullAnalysis}{G}
    \State configurations $\gets$ \Call{GenerateConfigurations}{G}
    \State ranks
    \For{config $\gets$ \Call{GetUnlabeledConfig}{configurations, ranks}}
        \State \Call{DFSUntilInvariant}{config, ranks}
    \EndFor
\EndProcedure
\Statex
\Function{GetUnlabeledConfig}{configurations, ranks}
    \While{config $\in$ configurations}
        \If{config $\notin$ ranks}
            \State \Return config
        \EndIf
    \EndWhile
    \State \Return null
\EndFunction
\Statex
\Procedure{DFSUntilInvariant}{config, ranks}
    \State L $\gets$ 0
    \State C $\gets$ 1
    \State M $\gets$ 0
    \While{config$\rq$ $\in$ \Call{GetProgramTransitions}{config}}
        \State L$\rq$, C$\rq$, M$\rq$ $\gets$ \Call{DFSUntilInvariant}{config$\rq$}
        \State L $\gets$ L + L$\rq$ + C$\rq$
        \State C $\gets$ C + C$\rq$
        \State M $\gets$ max(M, M$\rq$)
    \EndWhile
    \State ranks[config] $\gets \frac{\textnormal{L}}{\textnormal{C}}$
    \State \Return L, C, M
\EndProcedure
\end{algorithmic}
\end{algorithm}


\subsection{Exhaustive Full Analysis (FA)}
\label{sec:spa-fa}

Exhaustive FA is an SPA technique in which the entire configuration space of a program is exhaustively studied. For instance, $(\Delta+1)$-coloring on a star graph with 4 nodes yields $4 \cdot 2 \cdot 2 \cdot 2 = 32$ configurations, all of which are examined under FA. Algorithm~\ref{alg:spa-fa-analysis} presents a memory-efficient method for traversing the entire configuration space using Depth-First Search (DFS). We now introduce the following terminology used throughout our FA-based analysis:


\begin{definition}[Rank of the Configuration]
\label{def:rank}
The rank of a configuration in $(P, G)$ is the number of correct program transitions required for that configuration to reach a legitimate state.
\end{definition}

\begin{definition}[Max Rank and Average Rank]
\label{def:max-avg-rank}
The \emph{Max Rank} of a configuration, denoted $\mathrm{Rank}_{\max}$, is the maximum rank (Definition~\ref{def:rank}) observed across all possible transition sequences from that configuration to legitimate states of $(P, G)$. The \emph{Average Rank}, denoted $\mathrm{Rank}_{\mathrm{avg}}$, is the mean rank observed across the same set of transition sequences.
\end{definition}

\begin{figure}[htbp]
    \centering
    \includegraphics[width=\linewidth]{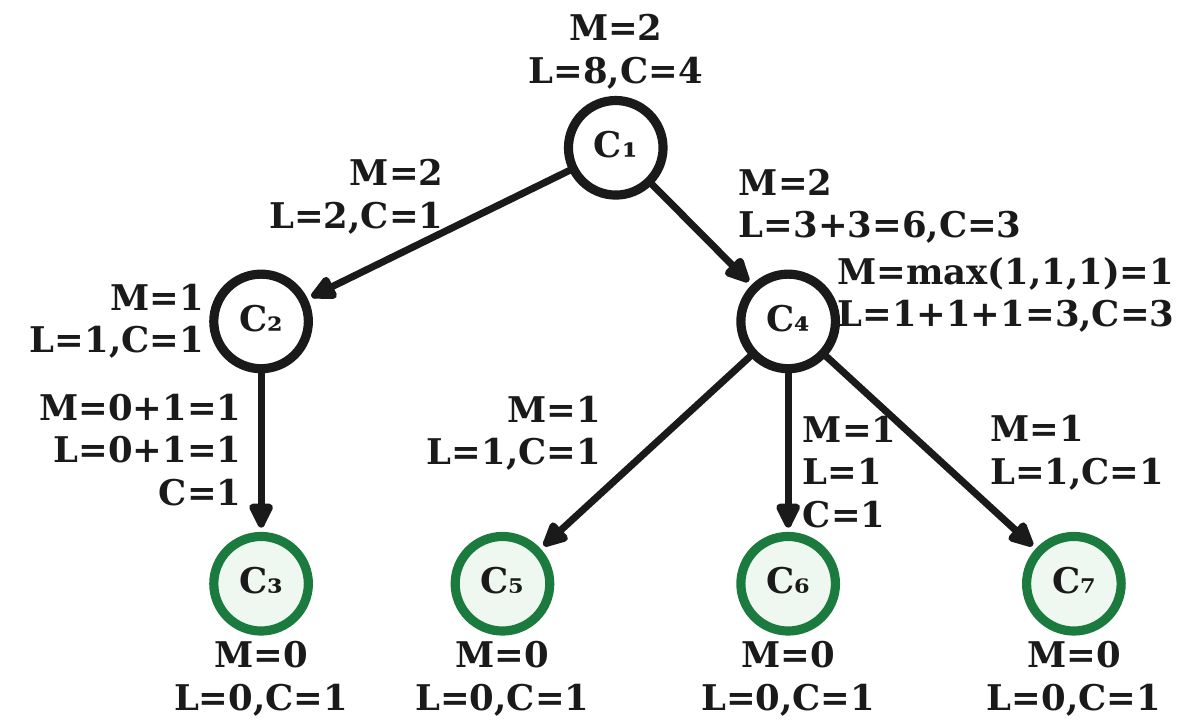}
    \caption{\small Max rank (M) and Average rank ($\frac{\textbf{L}}{\textbf{C}}$) calculations, starting from Invariant (M=0, L=0, C=1) and extending to the configuration at the root (M=2, L=8, C=4). \{C$_3$, C$_5$, C$_6$, C$_7$\} are the Invariant configurations.}
    \label{fig:rank-calc}
\end{figure}

Figure~\ref{fig:rank-calc} illustrates the rank calculation for the root configuration and the intermediate configurations that lead from the root configuration to the legitimate states (i.e., the stabilization program transitions). Here, $M$, $L$, and $C$ denote the Max Rank, Total Path Length, and Total Path Count, respectively, for any given configuration. The Average Rank is the ratio of $L$ to $C$. In the sections that follow, we focus our study on Average Rank, as it is the more complex of the two measures and the analysis applies equally to Max Rank.

\textbf{Rank Effect (RE).} RE quantifies the impact of a state transition in $(P, G)$ that results from reading stale values from neighboring nodes. Within our \cvf{} analysis model, RE is defined as the difference in rank between the source configuration $\gamma$ and the perturbed configuration $\gamma'$:
\[
RE(\gamma, \gamma \rq) = \text{rank}(\gamma \rq) - \text{rank}(\gamma)
\]
Intuitively, RE captures how many extra (or fewer) steps the perturbed state needs to reach the invariant, relative to the number of steps required under the correct (unperturbed) state transition.

\subsection{Conflicting Edges (CE) Method}
\label{sec:ce-method}
\begin{algorithm}[htbp]
\small
\caption{\small Static program analysis (rank computation) of a monotonically stabilizing program using the CE method.}
\label{alg:spa-ce-fa-analysis}
\begin{algorithmic}[1]
\Procedure{ConflictingEdgesBasedFullAnalysis}{G}
    \State configurations $\gets$ \Call{GenerateConfigurations}{G}
    \State ranks $\gets$ \Call{InitRanks}{configurations}
    \For{config $\gets$ \Call{GetUnrankedConfig}{configurations, ranks}}
        \State state $\gets$ \Call{GetEdgeState}{config} \label{line:ce-edge-state}
        \State state $\gets$ \Call{FilterNonConflictingEdges}{state} \label{line:ce-filter}
        \State \Call{RelabelState}{state} \label{line:ce-relabel-root}
        \State L, C $\gets$ \Call{TraverseCEState}{state}
        \State ranks[config] $\gets$ $\frac{\textnormal{L}}{\textnormal{C}}$ \label{line:ce-rank}
    \EndFor
\EndProcedure
\Statex
\Function{TraverseCEState}{state}
    \If{state $= \emptyset$} \label{line:ce-base}
        \State \Return (0, 1)
    \EndIf
    \If{state $\in$ CACHE} \label{line:ce-cache-lookup}
        \State \Return CACHE[state]
    \EndIf
    \State L $\gets$ 0
    \State C $\gets$ 0
    \For{node $\in$ \Call{GetUniqueNodeIds}{state}}
        \If{\Call{NodeIsEligible}{state, node}} \label{line:ce-eligible}
            \State state$\rq$ $\gets$ state $\setminus \{e \in$ state $\mid$ node $\in e\}$ \label{line:ce-branch} \Comment{Branching: eliminate edges incident to node.}
            \State \Call{RelabelState}{state$\rq$} \label{line:ce-relabel} \Comment{Relabeling of nodes.}
            \State L$\rq$, C$\rq$ $\gets$ \Call{TraverseCEState}{state$\rq$}
            \State L $\gets$ L + L$\rq$ + C$\rq$ \label{line:ce-accumulate}
            \State C $\gets$ C + C$\rq$
        \EndIf
    \EndFor
    \State CACHE[state] $\gets$ (L, C) \label{line:ce-cache-store}
    \State \Return (L, C)
\EndFunction
\Statex
\Procedure{RelabelState}{state}
    \State nodes $\gets$ \Call{GetUniqueNodeIds}{state}
    \State \Call{Sort}{nodes}
    \State newNodeIds $\gets$ \Call{MapNodeIdsSerially}{nodes}
    \For{edge $\in$ state}
        \State edge[0] $\gets$ newNodeIds[edge[0]]
        \State edge[1] $\gets$ newNodeIds[edge[1]]
    \EndFor
\EndProcedure
\end{algorithmic}
\end{algorithm}

Unlike FA, which exhaustively explores the configuration space, CE computes configuration ranks by working directly on conflicting edges (Algorithm~\ref{alg:spa-ce-fa-analysis}). CE exploits a structural symmetry in the state space of monotonically stabilizing programs (Definition~\ref{def:monotonic-stabilization}), of which GC is an instance. In such programs, a node needs only a single correct move to stabilize its state, after which it is never enabled again. For each configuration, CE therefore extracts its Edge State (Definition~\ref{def:invariant-edge-state}) and retains only the conflicting edges (Lines~\ref{line:ce-edge-state}--\ref{line:ce-filter}). \textproc{TraverseCEState} then explores the possible moves recursively: for each eligible node, it branches by eliminating all conflicting edges incident to that node, since the node's single correct move resolves them permanently (Branching and Elimination, Definition~\ref{def:branch-elim}; Lines~\ref{line:ce-eligible}--\ref{line:ce-branch}). The recursion ends when no conflicting edges remain (Line~\ref{line:ce-base}). Along the way, CE accumulates the number $C$ of stabilizing paths and their total length $L$, adding one move to every path through each branch (Line~\ref{line:ce-accumulate}), and the rank of the configuration is $L/C$ (Line~\ref{line:ce-rank}).

CE applies \textproc{RelabelState} (Definition~\ref{def:relabeling}) to every state it visits (Lines~\ref{line:ce-relabel-root} and~\ref{line:ce-relabel}), mapping the remaining node identifiers, in sorted order, to consecutive integers. Moves at different nodes can then lead to states whose remaining conflicting edges are identical after relabeling. In CE, the rank of a graph is determined by its structure and its Edge State rather than by its particular node values: every assignment of node values that satisfies the same Edge State yields the same rank. Such states therefore share a single CACHE entry, so CE computes their rank once and reuses it rather than recomputing it (Lines~\ref{line:ce-cache-lookup} and~\ref{line:ce-cache-store}).

CE does not, however, detect every pair of graphs with isomorphic structure. Because \textproc{RelabelState} only compacts node identifiers while preserving their order, it identifies only those pairs that become identical after relabeling; pairs that are isomorphic under any other correspondence of node identifiers receive separate CACHE entries. Detecting all such pairs would require an isomorphism test at every step. The Weisfeiler-Leman (WL) test~\cite{weisfeiler1968reduction} alone is unsuitable, because it fails to distinguish certain non-isomorphic graphs~\cite{cai1992optimal} and could therefore merge graphs with different ranks. Exact canonical labeling~\cite{mckay2014practical} avoids this, but at the cost of computing a canonical form for each graph. Exploring this trade-off is a natural direction for further optimizing CE.

The remaining functions in Algorithm~\ref{alg:spa-ce-fa-analysis} are defined as follows. \textproc{GenerateConfigurations} enumerates every configuration of $G$, i.e., every assignment of domain values to its nodes, and \textproc{InitRanks} creates a rank table in which every configuration is initially unranked. \textproc{GetUnrankedConfig} returns the next configuration whose rank has not yet been computed. \textproc{GetEdgeState} computes the Edge State of a configuration (Definition~\ref{def:invariant-edge-state}), marking each edge as conflicting or non-conflicting, and \textproc{FilterNonConflictingEdges} removes the non-conflicting edges, so that a state is represented by its set of conflicting edges. \textproc{GetUniqueNodeIds} returns the identifiers of all nodes that are endpoints of at least one edge in the state, and \textproc{NodeIsEligible} returns whether a node is enabled in the state; in GC, this holds for every endpoint of a conflicting edge. Within \textproc{RelabelState}, \textproc{Sort} and \textproc{MapNodeIdsSerially} implement this mapping, assigning the $i$-th identifier in ascending order to $i$, starting from $0$. Finally, CACHE is a global map from relabeled states to their computed pairs $(L, C)$, shared across all configurations.

A detailed workflow example is provided in the appendix (Figure~\ref{fig:ce-alg-execution}). Compared to exhaustive FA (Section~\ref{sec:spa-fa}), CE is substantially more memory-efficient and modestly faster, as demonstrated in Section~\ref{sec:results-ce-fa}.

\subsubsection{Proof of Correctness}
\begin{definition}[Monotonic Stabilization]
\label{def:monotonic-stabilization}
A self-stabilizing system is \emph{monotonically stabilizing} if, following a transient fault, every process changes its output at most once before the system reaches a legitimate state~\cite{yamauchi2010monotonic, faghih2018automated}. Equivalently, recovery proceeds without oscillation: each process moves directly from its post-fault output to its final, legitimate output, and the system as a whole progresses monotonically closer to a legitimate state at every step.
\end{definition}

\begin{definition}[Edge State and Invariant Edge State]
\label{def:invariant-edge-state}
Let $c$ be a configuration of $(\mathrm{GC}, G)$ with $G=(V,E)$, and let $c(v)$
denote the color of node $v$ in $c$. An edge $(u,v) \in E$ is
\emph{conflicting} in $c$ if $c(u) = c(v)$. The \emph{edge state} of $c$ is
the set of its conflicting edges,
\[
    s(c) = \{(u,v) \in E : c(u) = c(v)\} \subseteq E.
\]
An edge state $s$ is an \emph{invariant edge state} if $s = \emptyset$, i.e.,
no conflicting edges remain. Since every configuration with no conflicting
edges is legitimate, an invariant edge state has rank $0$ by
Definition~\ref{def:rank}. For an edge state $s$, the \emph{total path length} $L(s)$ and \emph{total path count} $C(s)$ are the sum of the lengths and the number of all transition sequences (paths) from $s$ to $\emptyset$.
\end{definition}

\begin{definition}[Branching and Elimination of Eligible Nodes]
\label{def:branch-elim}
Given an edge state $s$ with an eligible node $N_v$ (i.e., a node with at least one enabled action), the algorithm \emph{branches} by generating the successor state yielded by executing that action at $N_v$. By Definition~\ref{def:monotonic-stabilization}, this successor state \emph{eliminates} the edges associated with $N_v$, since a monotonically stabilizing algorithm stabilizes the node with at most one action on it.
\end{definition}

\begin{definition}[Relabeling of the Subgraph]
\label{def:relabeling}
Given an edge state $s'$, the relabeling of the subgraph relabels the node ids in serial order such that the existing node ids are as close as possible to the relabeled node ids in the subgraph. Since the subgraphs are structurally identical, the Rank (Definition~\ref{def:rank}) of the original subgraph and that of the relabeled subgraph are the same.
\end{definition}


\begin{lemma}[Correctness of CE-based Rank Computation]
\label{lem:ce-correctness}
For any edge state $s$ in $(P, G)$, the CE-based method correctly computes the Total Path Length $L$ and Total Path Count $C$ of $s$.
\end{lemma}

\begin{proof}
We proceed by induction on the number of transitions from $s$ to an invariant edge state (Definition~\ref{def:invariant-edge-state}). In the base case, $s$ is itself an invariant edge state, i.e., $s = \emptyset$, for which $L = 0$ and $C = 1$ by definition. For the inductive step, assume the claim holds for all edge states reachable from $s$ in fewer transitions than $s$ requires. The CE-based method comprises two steps, applied in order: \emph{Branching and Elimination of Eligible Nodes} (Definition~\ref{def:branch-elim}) and \emph{Relabeling of the Subgraph} (Definition~\ref{def:relabeling}). For each eligible node in $s$ (i.e., a node incident to one or more conflicting edges), branching and elimination correctly model the state transition of $s$ when a move is made by that eligible node, producing one successor per eligible node, each with strictly fewer conflicting edges than $s$. Relabeling, applied on top of each successor, introduces no side effects and thus allows the rank to be computed correctly for $s$ without altering the correctness established by the inductive hypothesis. Since $L$ and $C$ for $s$ are computed by aggregating over these successors, the claim holds for $s$, completing the induction.
\end{proof}




\subsection{Machine Learning (ML) Method}
\label{sec:ml-method}
\begin{algorithm}[htbp]
\small
\caption{\small ML-based \cvf{} analysis across rank and centrality attributes (degree, clustering, closeness, betweenness), stratified into high- and low-valued regimes. $S$ is the number of sampled configurations per regime; $C$ is the number of perturbed configurations sampled per configuration in $S$.}
\label{alg:ml-based-cvf-analysis}
\begin{algorithmic}[1]
\Procedure{MLBasedCVFAnalysis}{MLModel, Attribute, S, C}
    \State HiValConfigs, LoValConfigs $\gets$ \Call{GetHiLoValConfigs}{Attribute, S}
    \State HiValPertConfigs $\gets$ \Call{GetPerturbedConfigs}{HiValConfigs, C}
    \State LoValPertConfigs $\gets$ \Call{GetPerturbedConfigs}{LoValConfigs, C}
    \State HiREDist $\gets$ $\emptyset$
    \State LoREDist $\gets$ $\emptyset$
    \State \Call{ComputeRE}{MLModel, HiValConfigs, HiValPertConfigs, HiREDist}
    \State \Call{ComputeRE}{MLModel, LoValConfigs, LoValPertConfigs, LoREDist}
\EndProcedure
\Procedure{ComputeRE}{MLModel, Configs, PertConfigs, Dist}
    \For{config $\in$ Configs}
        \State rank $\gets$ MLModel(config)
        \For{config$\rq$ $\in$ PertConfigs}
            \State rank$\rq$ $\gets$ MLModel(config$\rq$)
            \State Dist $\gets$ Dist $\cup$ (rank$\rq$ - rank)
        \EndFor
    \EndFor
\EndProcedure
\end{algorithmic}
\end{algorithm}

An efficient way of computing the rank of a configuration in $(P, G)$ without exhaustively computing all the ranks is to use ML methods to learn the pattern of configurations in $(P, G)$ to get their ranks. In our study, we design a black-box ML model that takes as input a configuration alongside other precomputed attributes to predict the rank of the configuration. Since training data (which requires true rank labels) can only feasibly be generated for smaller graphs, it is important to consider how the model generalizes to the larger graphs used at inference time. Once we can compute rank efficiently for any configuration, we can analyze results based on sampled configurations in the configuration space. We leverage this ML-based approach to analyze faults in terms of \textit{when} and \textit{where} they occur, by examining sampled configurations in high- and low-valued regimes of each attribute. Algorithm~\ref{alg:ml-based-cvf-analysis} illustrates our ML-based approach for \cvf{} analysis across the evaluation attributes. \Call{GetHiLoValConfigs}{} first identifies the high- and low-valued regimes for a given attribute, then samples $S$ configurations from each regime.



%% file: experiment-setting.tex
\subsection{Graph Datasets}
\textbf{Graph families.} We evaluate on four structurally distinct graph
families, common to both the FA/CE experiments and the ML-based method:
complete graphs (\emph{K}), star graphs (\emph{St}), powerlaw cluster
graphs (\emph{Pl}), and path graphs (\emph{Pa}). This diversity lets us
assess whether the learned rank predictor generalizes across topologies
with markedly different degree distributions and diameters, rather than
overfitting to a single structural regime.

\textbf{In-Distribution (ID) and Out-of-Distribution (OOD) datasets.}
For the ML-based method specifically, because FA suffers from state-space
explosion, exhaustive enumeration is intractable beyond small graphs. We
therefore train the model on an ID dataset restricted to graphs within a
limited node-count range, and evaluate extrapolation on an OOD dataset
consisting of larger graphs whose node counts were never seen during
training. The OOD set serves as our validation set, and we track its
metrics alongside the ID training metrics throughout training to monitor
generalization rather than memorization.

\textbf{ML Features.} Also specific to the ML-based method, features used
to train the rank predictor follow a three-level taxonomy, from granular
to coarse. \emph{Node features} are local to a single node in a given
configuration: structural properties (degree, clustering, closeness,
betweenness centrality) and program state (e.g., \emph{is\_in},
\emph{can\_request} for MIS). \emph{Configuration features} aggregate
state across all nodes at a fixed configuration (e.g.,
\emph{multi\_gh\_max\_norm}, \emph{eligible\_frac}). \emph{Graph features}
are global, configuration-independent properties of $(P,G)$ (e.g., the
adjacency matrix), fixed regardless of system state.

\subsection{Measurement Protocol}
The following measurement protocol applies specifically to the FA and CE
experiments.

\begin{definition}[Comprehensive Phase]
\label{def:comprehensive-phase}
For a program-graph pair $(P, G)$, the \emph{comprehensive phase} of rank
computation, $T_{\text{comp}}(P,G)$, denotes the total wall-clock cost of
(i) allocating memory to store the rank of every configuration in the
configuration space of $(P,G)$, (ii) computing the rank of each
configuration, and (iii) writing each computed rank into its allocated
location.
\end{definition}

\begin{definition}[Core Phase]
\label{def:core-phase}
The \emph{core phase}, $T_{\text{core}}(P, G) \subseteq T_{\text{comp}}(P,G)$,
denotes the cost of step (ii) alone, isolated from allocation and storage
overhead. Under \emph{exhaustive full analysis (FA)}, $T_{\text{core}}(P,G)$
is the cost of a depth-first traversal of the full configuration space of
$(P,G)$. Under \emph{CE-based full analysis (CE)}, $T_{\text{core}}(P,G)$
is the cost of enumerating all conflicting-edge states reachable in
$(P,G)$.
\end{definition}










\subsection{ML Model Architecture}
Graph Neural Networks (GNNs) are the standard model architecture for learning over graph-structured data, following a general \emph{message-passing} scheme: at each layer, every node aggregates information from its neighbors, then combines this aggregated message with its own previous representation to produce an updated representation. Stacking $L$ such layers allows a node to incorporate information from nodes up to $L$ hops away. Different GNN variants are distinguished primarily by their choice of aggregation and combination functions, which in turn determines their expressive power. In recent years, several such architectures have emerged, including Graph Convolutional Networks (GCN)~\cite{kipf2017semisupervisedclassificationgraphconvolutional}, GraphSAGE~\cite{hamilton2018inductiverepresentationlearninglarge}, Graph Attention Networks (GAT)~\cite{veličković2018graphattentionnetworks}, and Graph Isomorphism Networks (GIN)~\cite{xu2018how}.

For our study, we consider GCN- and GIN-based networks: GIN as our primary architecture for its expressivity, since it is provably as powerful as the Weisfeiler-Lehman (WL) test~\cite{weisfeiler1968reduction} at distinguishing non-isomorphic graphs~\cite{xu2018how}, and GCN (following the standard formulation of~\cite{kipf2017semisupervisedclassificationgraphconvolutional}) as a point of comparison in our architecture ablation (Section~\ref{sec:ablation-ml-mis}), where its mean-based aggregation, not injective, and therefore strictly less expressive, provides a natural lower bound on model expressivity. We describe our primary GIN formulation below, where $h_v^{(l)}$ denotes the representation of node $v$ at layer $l$ and $\mathcal{N}(v)$ its neighborhood.


\noindent\textbf{GIN}~\cite{xu2018how} updates each node representation via a learned, injective aggregation:
\begin{equation}
    h_v^{(l)} = \mathrm{MLP}^{(l)}\left( (1+\epsilon^{(l)}) \, h_v^{(l-1)} + \sum_{u \in \mathcal{N}(v)} h_u^{(l-1)} \right)
\end{equation}
where $\epsilon^{(l)}$ is a learnable (or fixed) scalar and $\mathrm{MLP}^{(l)}$ is a multi-layer perceptron. Because summation over neighbors, combined with an MLP, can approximate any injective function on multisets, GIN can, in principle, distinguish any pair of graphs that the WL test can distinguish.

\subsection{ML Training and Evaluation Setup}

\begin{table}[htbp]
\centering
\small
\setlength{\tabcolsep}{6pt}
\caption{Rank distribution by split, training summary, and hardware. $n$ is the number of configurations employed for training/validation.}
\label{tab:rank-dist}
\begin{tabular*}{\linewidth}{@{\extracolsep{\fill}} lrrrr}
\toprule
\textbf{Split} & $n$ & \textbf{Mean} & \textbf{Median} & \textbf{Max} \\
\midrule
ID Train      & 36{,}015 & 5.65  & 5.00  & 16.33 \\
ID Validation & 6{,}355  & 5.65  & 5.00  & 15.71 \\
OOD Validation & 39{,}600 & 12.60 & 12.51 & 30.02 \\
\midrule
\multicolumn{5}{c}{Model Parameters: 34{,}499 \quad Epochs: 50 \quad Train time: 184s } \\
\multicolumn{5}{c}{Loss Function: Mean Absolute Error (MAE)} \\
\multicolumn{5}{c}{GPU: RTX 4090 \quad CPU: Intel i7 (13th Gen) \quad RAM: 128\,GB} \\
\bottomrule
\end{tabular*}
\end{table}

Table~\ref{tab:rank-dist} presents the training setup for our ML model designed for predicting rank for a given configuration in $(MIS, G)$. ID graphs (train/validation) range from 3 to 8 nodes, while the OOD split ranges from 9 to 13 nodes; mean and median rank roughly double between the two, reflecting this size shift. Despite the modest training footprint, the model is evaluated against an OOD state space several orders of magnitude larger than what FA can exhaustively enumerate at comparable graph sizes, underscoring the scalability motivation for the ML-based approach over exhaustive analysis.

We evaluate the trained ML model along two complementary axes. The first is rank-based \cvf{} analysis (the \emph{when} aspect), in which high- and low-rank configurations are sampled from the configuration space as described in Section~\ref{sec:rank-sampling}; the rank effect distribution is then generated by applying a limited number of single-step perturbations to each sampled configuration and computing the resulting rank effect. The second is node-based \cvf{} analysis (the \emph{where} aspect), in which we use four centrality measures: degree, clustering, closeness, and betweenness, described in Section~\ref{sec:centrality-measures}, and sample high- and low-valued nodes for each; for a given source configuration, we perturb the sampled node and compute the resulting rank effect distribution.

\subsubsection{Sampling Low- and High-Rank Configurations}
\label{sec:rank-sampling}

To study how RE (Section~\ref{sec:spa-fa}) varies across the configuration space, we construct two contrasting pools of configurations for a given graph $G$. The \emph{low-rank} pool is built systematically: for every invariant configuration $\gamma \in I(P,G)$, we enumerate all single-node perturbations of $\gamma$ (perturbing exactly one node), and retain a perturbed configuration $\gamma'$ only if it is not itself invariant but has at least one enabled action transitioning directly back into $I(P,G)$, guaranteeing a rank of $\approx 1$.

The \emph{high-rank} pool is instead hand-constructed to cover several qualitatively distinct failure patterns: uniform extremes (all OUT, all IN, all WAIT), high concentrations of simultaneous WAIT nodes, dense IN--IN conflicts, alternating IN/OUT patterns, hub-concentrated messes, and random mixes weighted toward IN and WAIT states.

\subsubsection{Centrality Measures}
\label{sec:centrality-measures}
For node-based \cvf{} analysis, we characterize each node $v$ in a graph $G=(V,E)$, $n=|V|$, using four standard centrality measures.

\noindent\textbf{Degree centrality} measures a node's direct connectivity:
\begin{equation}
    C_D(v) = \frac{\deg(v)}{n-1}
\end{equation}
where $\deg(v)$ is the number of neighbors of $v$.

\noindent\textbf{Clustering coefficient} measures how interconnected a node's neighbors are:
\begin{equation}
    C_{cl}(v) = \frac{2 e_v}{k_v (k_v - 1)}
\end{equation}
where $k_v = \deg(v)$ and $e_v$ is the number of edges among $v$'s neighbors; $C_{cl}(v) = 0$ if $k_v < 2$.

\noindent\textbf{Closeness centrality} measures how close a node is, on average, to all other nodes:
\begin{equation}
    C_C(v) = \frac{n-1}{\sum_{u \in V \setminus \{v\}} d(v,u)}
\end{equation}
where $d(v,u)$ is the shortest-path distance between $v$ and $u$.

\noindent\textbf{Betweenness centrality} measures how often a node lies on shortest paths between other node pairs:
\begin{equation}
    C_B(v) = \sum_{\substack{s,t \in V \setminus \{v\} \\ s \neq t}} \frac{\sigma_{st}(v)}{\sigma_{st}}
\end{equation}
where $\sigma_{st}$ is the number of shortest paths between $s$ and $t$, and $\sigma_{st}(v)$ is the number of those passing through $v$.

For each measure, we identify the highest- and lowest-valued nodes as the high- and low-valued node samples used in node-based \cvf{} analysis (Section~\ref{sec:spa-fa}).

\subsection{Graph Size-Invariant Design for OOD Generalization}

\subsubsection{Multi Greedy Height}
\label{sec:multi-greedy-height}
The configuration features prefixed by \emph{multi\_gh} estimate a configuration's distance from the invariant, analogous to the rank function itself: zero at the invariant, and decreasing along legal transitions toward it. As a proxy for this distance, we use the \emph{height} of a greedy rollout to an invariant state, i.e.\ its path length, computed over 16 such rollouts (two deterministic, two greedy on \emph{violation\_score}/\emph{privilege\_count}, and twelve random). \emph{multi\_gh\_mean\_norm} and \emph{multi\_gh\_max\_norm} are the mean and maximum height over these 16 rollouts, respectively, each normalized by the number of nodes in the graph; the \emph{unnormalized} maximum height additionally serves as the Lyapunov base for the model's residual prediction target (\texttt{lyap\_base=max}; Section~\ref{sec:scaled-residual}). Two of the greedy trials probe different notions of ``distance from done'': one maximizes \emph{privilege\_count}, the number of nodes with a currently enabled move (zero only at the invariant), while the other maximizes \emph{violation\_score}, a count of adjacent IN--IN conflicts, WAIT nodes, and uncovered OUT nodes, together steering rollouts toward busier or more structurally unresolved configurations, respectively.

\subsubsection{Scaled Residual}
\label{sec:scaled-residual}

Rather than predicting rank directly, the model predicts a \emph{scaled residual} below a Lyapunov upper bound, computed as follows.

\noindent \textbf{Lyapunov base} $\mathrm{lyap}$ (\texttt{lyap\_base=max}): the raw maximum greedy height to an invariant configuration (Section~\ref{sec:multi-greedy-height}), i.e.\ \emph{multi\_gh\_max\_norm} before normalization by $n$.

\noindent \textbf{Scale}: $\mathrm{scale} = \max(\mathrm{diameter}, n-1, 1)$, where $\mathrm{diameter}$ is the diameter of the graph and $n$ is the number of nodes, so larger graphs receive a proportionally larger denominator.

\noindent \textbf{Training target}:
\begin{equation}
    y = \frac{\mathrm{lyap} - \mathrm{rank}}{\mathrm{scale}}
\end{equation}
The gap between the Lyapunov bound and the true rank, expressed in units of graph scale rather than absolute rank.

\noindent \textbf{Decoding at inference}:
\begin{equation}
    \hat{\mathrm{rank}} = \mathrm{clamp}\left(\mathrm{lyap} - \hat{y} \times \mathrm{scale},\ 0\right)
\end{equation}
Since $\mathrm{rank}$ (Definition~\ref{def:rank}) can never be negative, $\mathrm{clamp}(x, 0) = \max(x, 0)$ floors the raw decoded value at 0, guaranteeing a valid rank even when prediction error would otherwise yield a negative result.

This scaling is what allows the model to share weights across graph sizes: under a raw residual ($y = \mathrm{lyap} - \mathrm{rank}$), the gap grows directly with $n$, making a fixed prediction error correspond to very different relative errors on small versus large graphs. The scaled residual instead expresses the gap as a size-relative quantity. For example, a residual near 1 means the gap is about as large as the graph's diameter, whether $n=5$ or $n=12$.

%% file: result.tex
\subsection{CE-based  analysis on ($\Delta$ + 1)-GC}
\label{sec:results-ce-fa}
\begin{figure}[htbp] 
    \centering
    \begin{subfigure}[b]{\linewidth}
        \includegraphics[width=\linewidth]{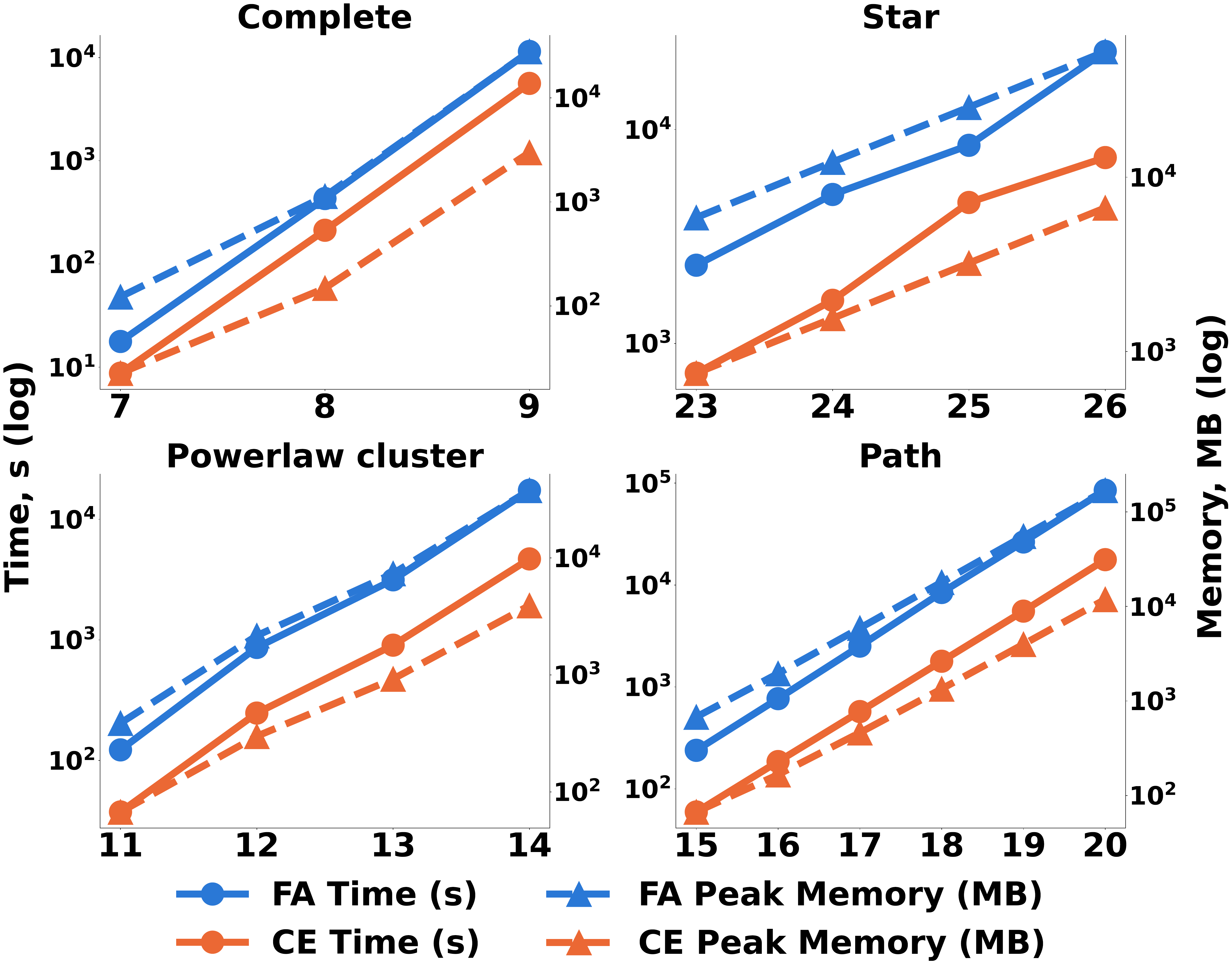}
        \caption{\small Comparison of total time (seconds) and peak memory (MB) between FA and CE for rank (Comprehensive phase from definition~\ref{def:comprehensive-phase}) computation, shown across four graph families.}
        \label{fig:placeholder1}
    \end{subfigure}
    \newline \newline
    \begin{subfigure}[b]{\linewidth}
        \includegraphics[width=\linewidth]{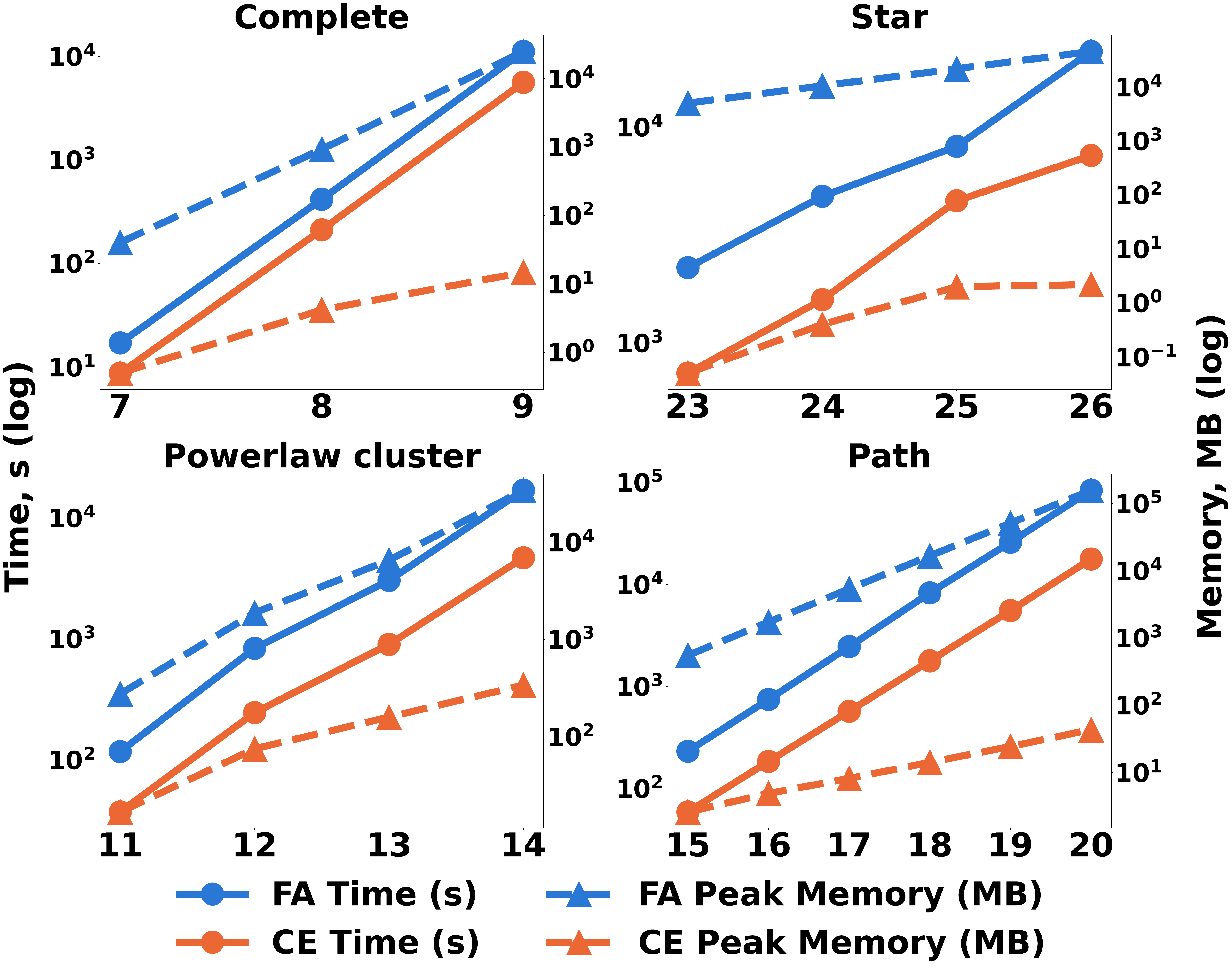}
         \caption{\small Comparison of total time (seconds) and peak memory (MB) between FA and CE for rank (Core phase from definition~\ref{def:core-phase}) computation, shown across four graph families.}
        \label{fig:placeholder2}
    \end{subfigure}
    \caption{\small Total time and peak memory for computing rank via exhaustive FA versus the CE-based method on the $(\Delta+1)$-GC program, across graph families and sizes.}
    \label{fig:ce-results}
\end{figure}

\input{cache_hit_miss_table}

Figure~\ref{fig:ce-results} compares the performance of exhaustive FA and CE-based FA for the GC program across different graph types and sizes. CE-based FA proves more efficient than exhaustive FA in both time and memory, with its memory advantage being most pronounced during the core phase (Definition~\ref{def:core-phase}). As graph size increases, the memory gap between exhaustive FA and CE-based FA widens superlinearly rather than growing linearly. This trend is supported by the cache hit rate, which likewise increases superlinearly with graph size.

Among the four graph types studied, powerlaw cluster graphs incurred the most cache misses, despite having a substantially smaller state space than the largest path and star instances tested (Table~\ref{tab:ce-cache-hits-misses}). Star graphs, by contrast, incurred the fewest misses in absolute and relative terms: misses remained pinned at essentially the node count even as the number of configurations grew substantially larger, yielding a cache hit rate that rounds to 100.00\% for the larger instances tested. This is consistent with the underlying graph structure: powerlaw cluster graphs, being highly heterogeneous, produce a greater number of structurally distinct subgraphs, whereas star graphs, owing to their highly symmetric structure, produce comparatively few distinct subgraphs~\cite{zenil2014correlation}.

Since rank effect calculation proceeds identically under FA and CE as described in Section~\ref{sec:spa-fa}, we omit it from our \cvf{} analysis.

\subsection{Ablation Study of ML-based method on MIS}
\label{sec:ablation-ml-mis}

\begin{table}[htbp]
\centering
\small
\caption{\small Group ablation on the 8-feature set. $n$ is the number of features remaining after removal.}
\label{tab:group-ablation}
\begin{tabular*}{\linewidth}{@{\extracolsep{\fill}} p{4.5cm}ccc}
\toprule
\textbf{Configuration} & $n$ & \textbf{ID MAE} & \textbf{OOD MAE} \\
\midrule
\textbf{Full model (8 features)} & \textbf{8} & \textbf{0.0756} & \textbf{0.2875} \\
\midrule
w/o \textbf{node group}\newline is\_in, is\_eligible, can\_request, can\_join, can\_leave & 3 & 0.1656 & 0.6567 \\
\addlinespace
w/o \textbf{configuration group}\newline multi\_gh\_max\_norm, multi\_gh\_mean\_norm & 6 & 0.1010 & 0.4322 \\
\addlinespace
w/o \textbf{eligible\_frac} & 7 & 0.0922 & 0.2920 \\
\bottomrule
\end{tabular*}
\end{table}

\textbf{Feature Importance.} Table~\ref{tab:group-ablation} presents an ablation study of the features used to train the $(MIS, G)$ ML model. Features are grouped from granular to coarse: node features (per-node, program \emph{MIS} related: \emph{is\_in}, \emph{can\_request}, \emph{can\_join}, \emph{can\_leave}, plus the generic feature \emph{is\_eligible}), configuration features (per-configuration, constant across nodes: \emph{multi\_gh\_max\_norm}, \emph{multi\_gh\_mean\_norm}, \emph{eligible\_frac}), and graph features (global, fixed across all configurations: the adjacency matrix of $G$).

Each node feature marks a specific fact about that node's current role: \emph{is\_in} indicates current set membership, \emph{is\_eligible} indicates whether any MIS rule is enabled, and \emph{can\_request}, \emph{can\_join}, and \emph{can\_leave} are the enabled-action bits for requesting membership (OUT $\to$ WAIT, no IN neighbor), joining the set (WAIT $\to$ IN, no IN neighbor and no lower-id WAIT neighbor), and leaving the set to resolve an ID-based conflict (IN $\to$ OUT, a lower-id IN neighbor exists). At the configuration level, \emph{eligible\_frac} gives the fraction of nodes with any enabled rule, zero only at the invariant.

\begin{table}[htbp]
\centering
\small
\caption{\small Ablation over GNN layer type, depth, hidden size, and pooling.}
\label{tab:arch-ablation}
\begin{tabular*}{\linewidth}{@{\extracolsep{\fill}} llcc}
\toprule
\textbf{Component} & \textbf{Setting} & \textbf{ID MAE} & \textbf{OOD MAE} \\
\midrule
\multirow{2}{*}{GNN layer}
 & GIN & \textbf{0.0756} & \textbf{0.2875} \\
 & GCN & 0.0785 & 0.3424 \\
\midrule
\multirow{3}{*}{Depth ($L$)}
 & 2 & 0.0733 & \textbf{0.2797} \\
 & 3 & \textbf{0.0673} & 0.3186 \\
 & 4 & 0.0745 & 0.3173 \\
\midrule
\multirow{3}{*}{Hidden dim}
 & 64  & 0.0791 & \textbf{0.3125} \\
 & 96  & \textbf{0.0673} & 0.3186 \\
 & 128 & 0.0738 & 0.3278 \\
\midrule
\multirow{3}{*}{Pool}
 & mean & \textbf{0.0673} & \textbf{0.3186} \\
 & max  & 0.0829 & 0.3378 \\
 & add  & 0.1060 & 0.4257 \\
\bottomrule
\end{tabular*}
\end{table}

\textbf{Model Architecture Ablation.} Table~\ref{tab:arch-ablation} presents the ablation of key components and hyperparameters, where depth ($L$) is the number of GNN layers. GIN outperforms GCN, consistent with its stronger expressivity under the WL test. Shallower depth and smaller hidden dimension generalize better OOD, suggesting reduced capacity aids robustness to distribution shift. Pooling has the largest effect: mean pooling's OOD advantage is consistent with add pooling's sensitivity to graph-size shift between splits.

\subsection{Evaluation of ML-based analysis on MIS}

Figure~\ref{fig:ml-rank-based-all} (rank-based, the \emph{when} aspect) and Figures~\ref{fig:ml-node-based-complete}, \ref{fig:ml-node-based-star}, \ref{fig:ml-node-based-powerlaw}, and \ref{fig:ml-node-based-path} (node-based, the \emph{where} aspect) present the ML-based \cvf{} analysis for the MIS program on four graph types: complete, star, powerlaw cluster, and path, each with $n=100$ nodes ($3^{100}$ configuration space for MIS). In Figure~\ref{fig:ml-rank-based-all}, it can be observed for all graph types that lower-rank configurations skew toward positive rank effects, while higher-rank configurations skew toward negative rank effects. This is consistent with rank being bounded below by the invariant (rank $= 0$): a low-rank configuration has comparatively little room to decrease further but substantial room to increase, so perturbations more often move it to a higher-rank configuration, while a high-rank configuration has the reverse asymmetry, with substantially more of the configuration space reachable at lower rank than higher, so perturbations more often decrease its rank. This indicates the ML model successfully learned the relative pattern of \cvf{} rank effects distinguishing higher-rank from lower-rank configurations.

In Figures~\ref{fig:ml-node-based-complete}, \ref{fig:ml-node-based-star}, \ref{fig:ml-node-based-powerlaw}, and \ref{fig:ml-node-based-path}, the ML model successfully learns the relative pattern of \cvf{} rank effects across nodes with differing centrality scores. This is most directly verifiable on structurally symmetric graphs: the complete graph is vertex-transitive, meaning every node is structurally equivalent to every other node, so all four centrality measures agree across nodes, and the model's predicted rank effect distribution is correspondingly uniform. The star graph, by contrast, has two structurally distinct groups of nodes, the hub and the leaves, and the model's predictions correctly reflect this two-group structure rather than treating all nodes as equivalent. The path graph offers an intermediate case: degree centrality varies only slightly between endpoint and internal nodes, and the model's predicted rank effects vary correspondingly little across the graph. For powerlaw cluster graphs, high- and low-valued nodes show clear differences in predicted rank effect under degree, closeness, and betweenness, matching the underlying structural differences in those measures; clustering coefficient is the exception, where high- and low-valued nodes show similar predicted effects despite structural differences in clustering, a discrepancy left for further investigation.

In both analyses, the ML model produced logically sound observations when comparing rank effects between the high- and low-valued extremes of a given attribute. However, it remains unclear whether the model's \emph{absolute} predicted rank effects are also accurate, since ground-truth labels are computationally intractable for a state space of $3^{100}$; verifying absolute accuracy at this scale is therefore out of scope for this paper. It can nonetheless be observed from the plots that predicted rank effects scale with graph size, even though the training graphs had substantially smaller rank effects, supporting our use of \emph{multi greedy height} (Section~\ref{sec:multi-greedy-height}) and \emph{scaled residual} (Section~\ref{sec:scaled-residual}).

\begin{figure}[htbp]
    \centering
    \includegraphics[width=\linewidth]{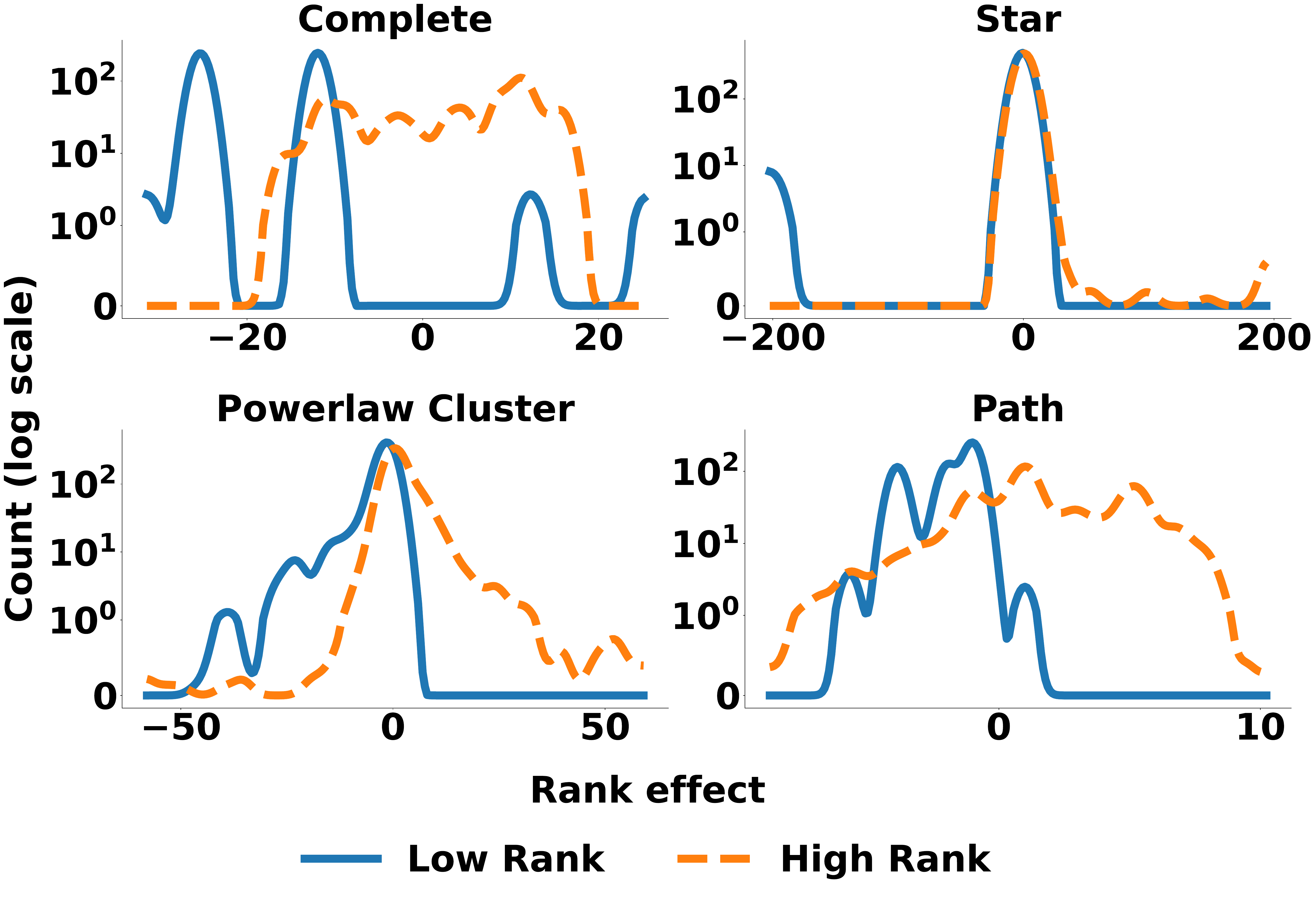}
    \caption{\small \Cvf{} analysis of configurations across different graph types in $(MIS, G_{N100})$ ($3^{100}$ configuration space), sampling 100 initial configurations $\times$ 50 perturbations each, per group (high/low-rank).}
    \label{fig:ml-rank-based-all}
\end{figure}

\begin{figure}[htbp]
    \centering
    \includegraphics[width=\linewidth]{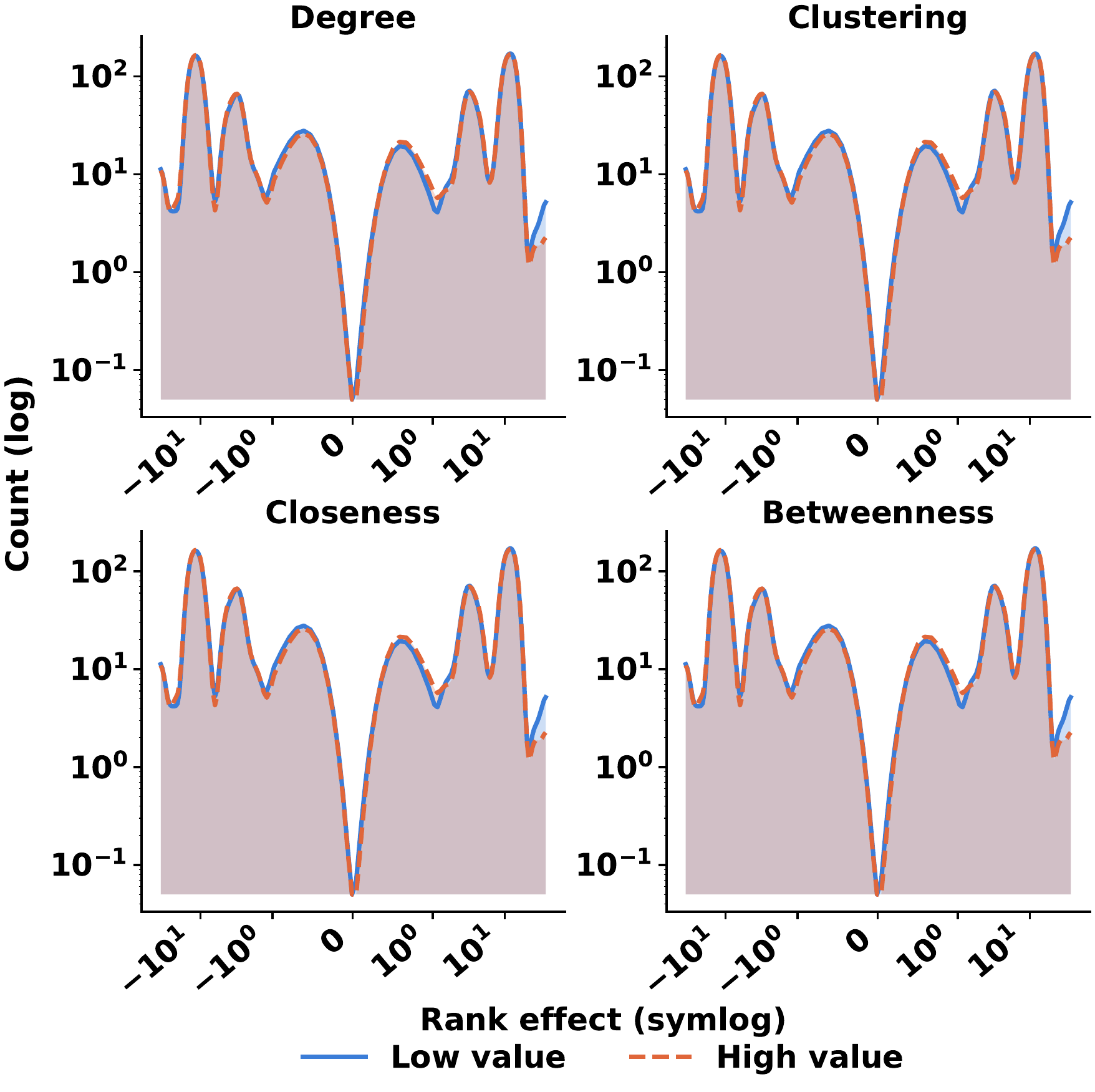}
    \caption{\small ML-based \cvf{} analysis for high- and low-value configurations across different centrality scores in \emph{(MIS, K\_N100)} ($3^{100}$ state space). Sampling 100 initial configs × 50 perturbations each, per group (high/low-value).}
    \label{fig:ml-node-based-complete}
\end{figure}

\begin{figure}[htbp]
    \centering
    \includegraphics[width=\linewidth]{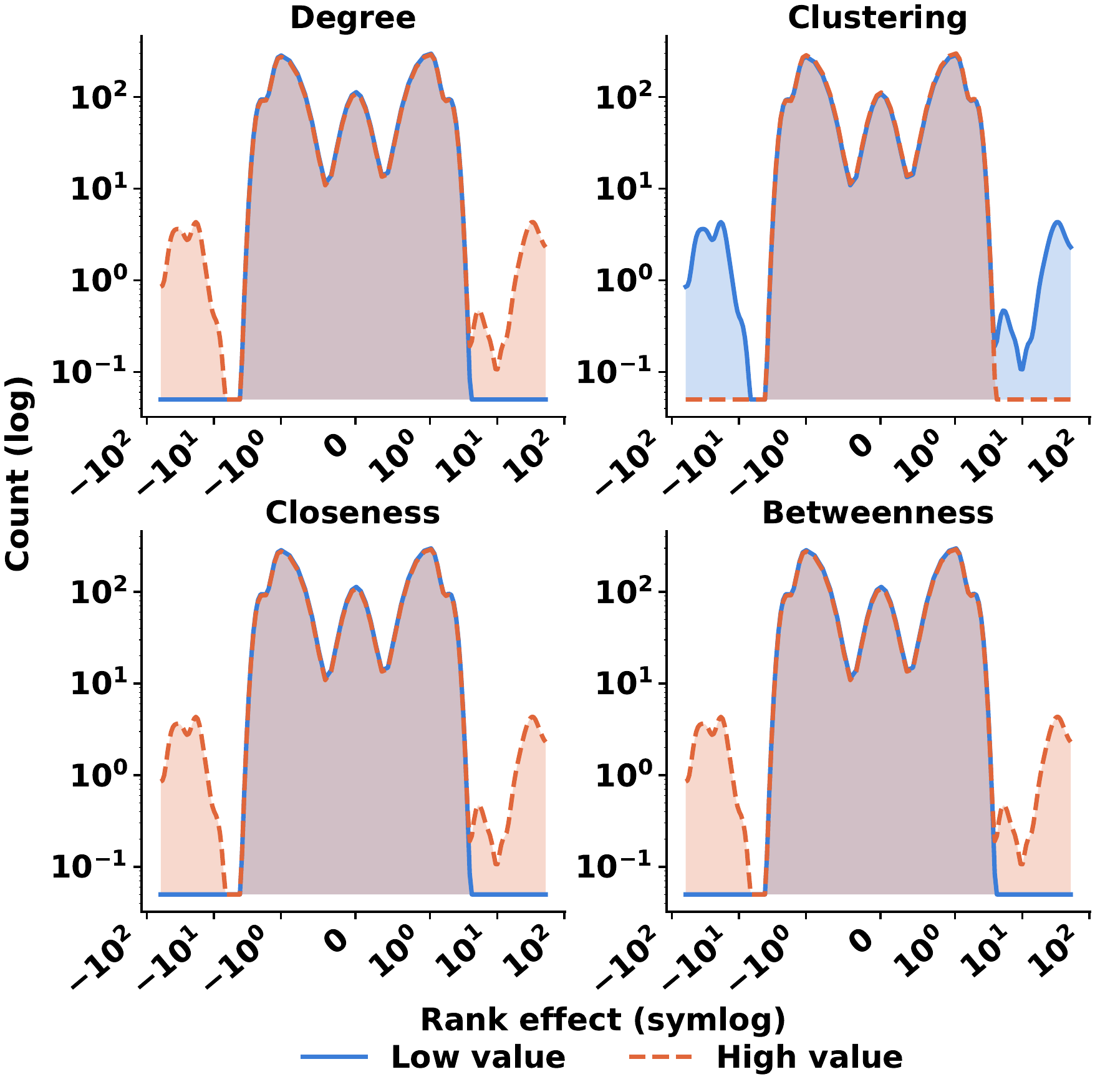}
    \caption{\small ML-based \cvf{} analysis for high- and low-value configurations across different centrality scores in \emph{(MIS, St\_N100)} ($3^{100}$ state space). Sampling 100 initial configs × 50 perturbations each, per group (high/low-value).}
    \label{fig:ml-node-based-star}
\end{figure}

\begin{figure}[htbp]
    \centering
    \includegraphics[width=\linewidth]{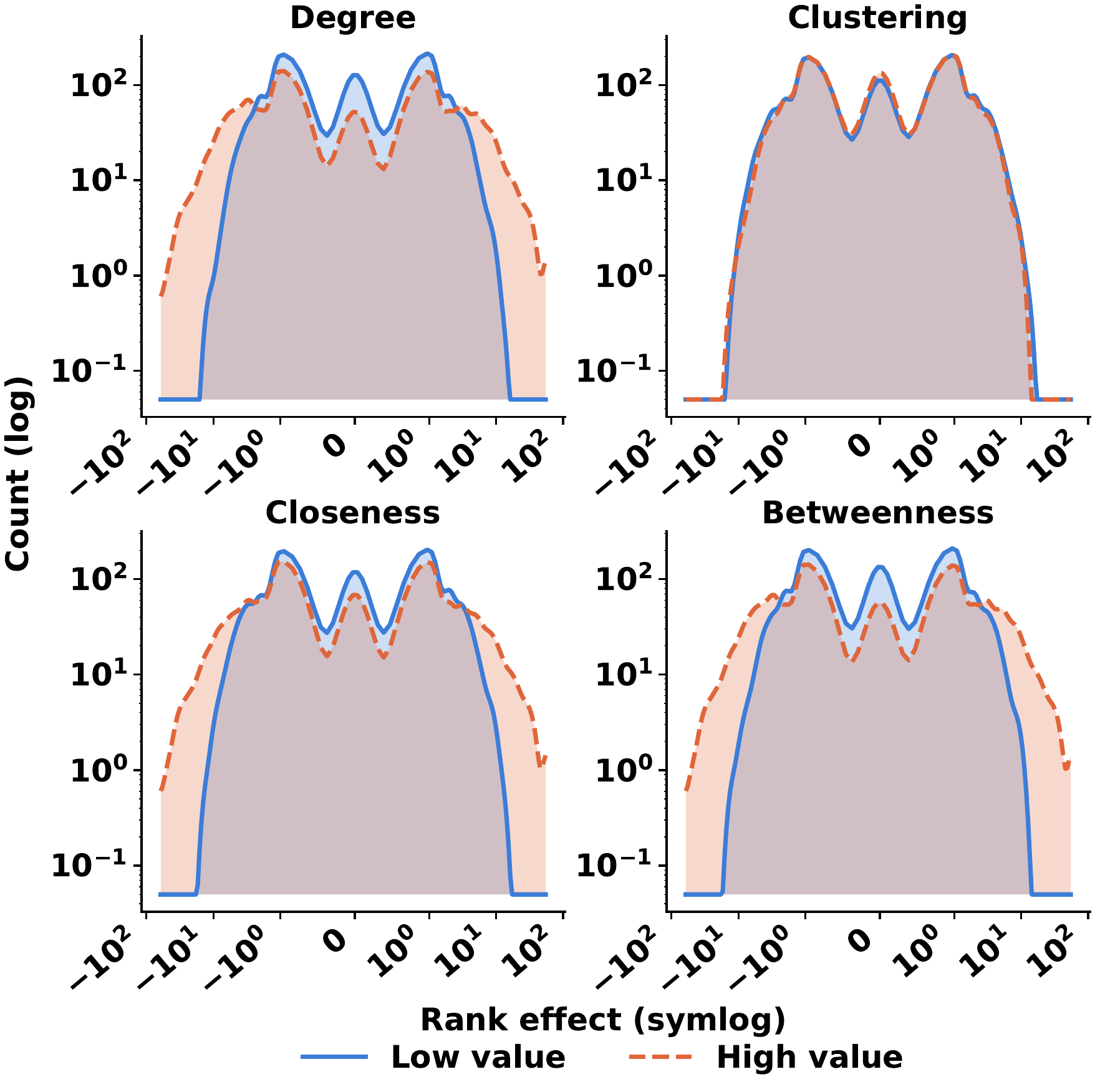}
    \caption{\small ML-based \cvf{} analysis for high- and low-value configurations across different centrality scores in \emph{(MIS, Pl\_N100)} ($3^{100}$ state space). Sampling 100 initial configs × 50 perturbations each (on the targeted node), per group (high/low-value).}
    \label{fig:ml-node-based-powerlaw}
\end{figure}

\begin{figure}[htbp]
    \centering
    \includegraphics[width=\linewidth]{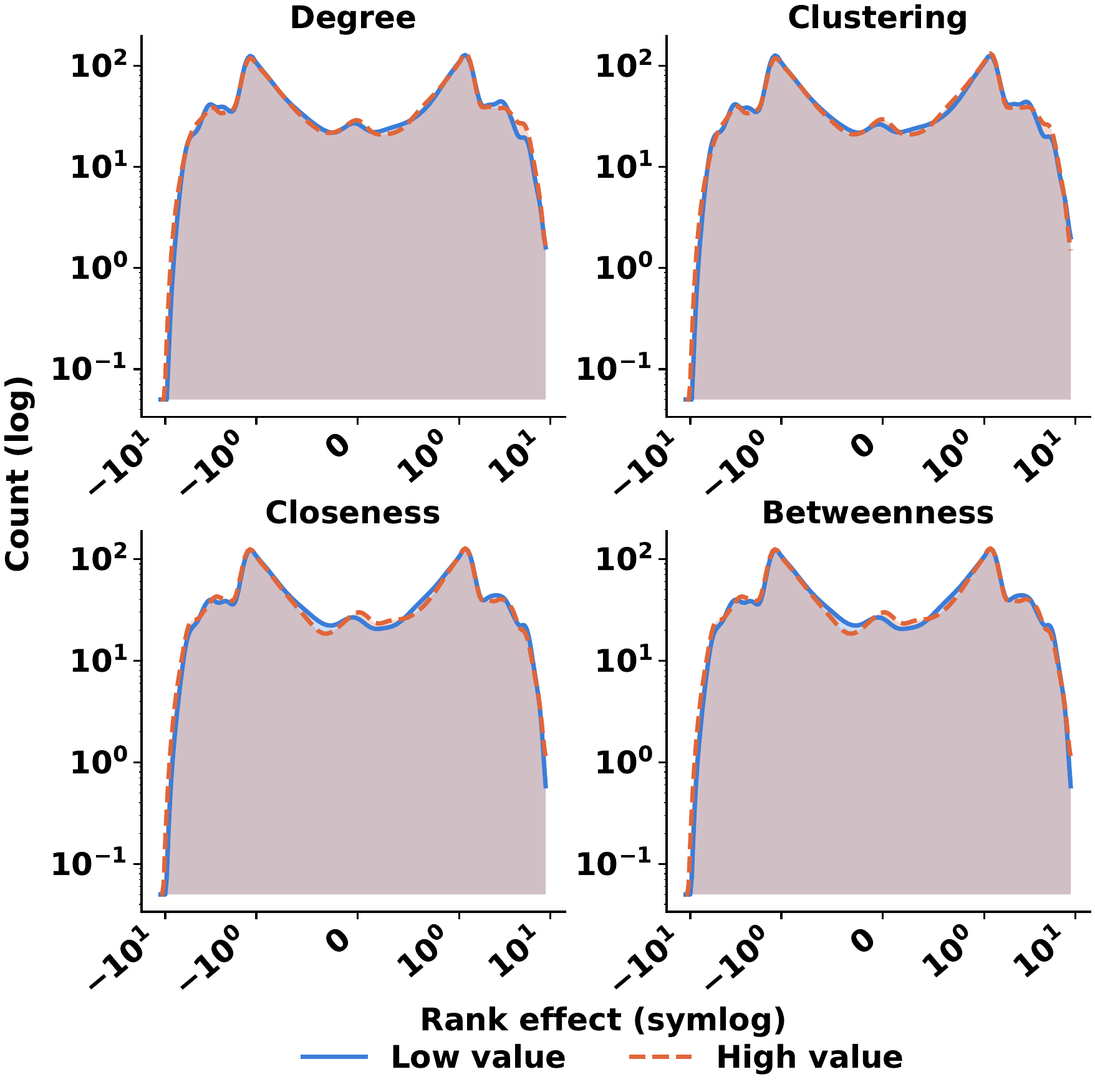}
    \caption{\small ML-based \cvf{} analysis for high- and low-value configurations across different centrality scores in \emph{(MIS, Pa\_N100)} ($3^{100}$ state space). Sampling 100 initial configs × 50 perturbations each (on the targeted node), per group (high/low-value).}
    \label{fig:ml-node-based-path}
\end{figure}


%% file: cache_hit_miss_table.tex
\begin{table}[htbp]
\centering
\footnotesize
\renewcommand{\arraystretch}{0.9}
\setlength{\tabcolsep}{4pt}
\caption{\small Cache hits and misses of edge states when employing the CE-based method for rank computation.}
\begin{tabular}{lcrrrr}
\hline
\textbf{Family} & \textbf{Graph} & \textbf{Configs} & \textbf{Cache Hits} & \textbf{Cache Misses} & \textbf{Hit Rate ($\approx$\%)} \\
\hline
K & n7  & 823543     & 824778     & 221   & 99.97 \\
K & n8  & 16777216   & 16783456   & 936   & 99.99 \\
K & n9  & 387420489  & 387454129  & 4361  & 100.00 \\
\hline
St & n23 & 96468992   & 96469244   & 23 & 100.00 \\
St & n24 & 201326592  & 201326867  & 24 & 100.00 \\
St & n25 & 419430400  & 419430699  & 25 & 100.00 \\
St & n26 & 872415232  & 872415556  & 26 & 100.00 \\
\hline
Pl & n11 & 4572288   & 4625212   & 6932   & 99.85 \\
Pl & n12 & 27216000  & 27460998  & 28593  & 99.90 \\
Pl & n13 & 97977600  & 98511962  & 57334  & 99.94 \\
Pl & n14 & 466560000 & 467705784 & 116214 & 99.98 \\
\hline
Pa & n15 & 6377292    & 6389514    & 987   & 99.99 \\
Pa & n16 & 19131876   & 19153248   & 1597  & 99.99 \\
Pa & n17 & 57395628   & 57432792   & 2584  & 100.00 \\
Pa & n18 & 172186884  & 172251197  & 4181  & 100.00 \\
Pa & n19 & 516560652  & 516671477  & 6765  & 100.00 \\
Pa & n20 & 1549681956 & 1549872220 & 10946 & 100.00 \\
\hline
\end{tabular}
\label{tab:ce-cache-hits-misses}
\end{table}

%% file: related-work.tex
\Cvf{}s have been studied in~\cite{Nguyen2019ICDCN,NguyenGuptaKulkarni23ICDCNAnalyzing} using exhaustive state space exploration and sampling methods. However, neither study provides analysis on large graphs or a scalable method for handling them.

Prior work has explored a range of techniques for improving OOD generalization in GNNs, including feature decorrelation~\cite{li2021oodgnnoutofdistributiongeneralizedgraph}, sparsity regularization~\cite{nerem2025gnn}, causal invariant learning~\cite{chen2022causally, wu2024causal}, data augmentation and graph structure extrapolation~\cite{li2023graphstructure}, and explicit size-generalization regularization~\cite{buffelli2022sizeshiftreg}. In our setting, we instead exploit multi-greedy height and scaled residual connections to achieve OOD generalization.

Model checking has been applied to verify self-stabilizing algorithms, using techniques such as symbolic state representation with OBDDs~\cite{tsuchiya2001symbolic} and weak-stabilization relaxations to improve scalability~\cite{jaubert2012scalable}. These approaches reduce the cost of exhaustive verification through symbolic encoding or relaxed fairness assumptions, but do not target \cvf{} analysis specifically, nor do they exploit the structural properties of the underlying graph to guide exploration. We instead propose a conflicting edges based state space exploration method that models the behavior of monotonically stabilizing programs on graphs efficiently by directly exploiting the structural properties of the underlying graph to guide exploration.

%% file: conclusion.tex
Distributed systems are inherently susceptible to a range of faults because they comprise multiple independently updated, replicated components rather than a single unified state. One such fault is \cvf{}, which occurs when a node acts on stale information about its neighbors. Self-stabilizing programs are tolerant of such faults within certain bounds; however, analyzing self-stabilizing graph programs is inherently non-trivial due to the size of the configuration space, which grows intractable even for small graphs (around 10 nodes). In this paper, we propose two methods for \cvf{} analysis that avoid the state-space explosion inherent to exhaustive exploration: a conflicting-edge (CE) based analysis that exploits the monotonic stabilization property, studied over the $(\Delta+1)$-Graph Coloring program, and a machine-learning-based analysis, studied over the Maximal Independent Set (MIS) program. The designed ML model was able to capture, in a computationally tractable way, the relative pattern between rank-effect distributions of configurations at the extremes of several attributes, including rank and node centrality, even on out-of-distribution graphs.

For future work, one direction is to study the \emph{absolute} accuracy of predicted rank effects on OOD graphs, which would provide stronger guarantees for \cvf{} analysis beyond the relative patterns validated in this paper. A second direction is combining CE-based analysis with ML to compute ranks more efficiently for monotonically stabilizing programs: configurations could be cached by their edge state, so that a rank is fetched directly from a matching cached edge state rather than recomputed via ML inference, substantially reducing resource use at scale. A third direction is extending the ML-based approach to other self-stabilizing programs beyond MIS, reusing the multi greedy height (Section~\ref{sec:multi-greedy-height}) and scaled residual (Section~\ref{sec:scaled-residual}) techniques for OOD modeling.

%% file: appendixes.tex

\section{Additional Figures}
\begin{strip}
    \centering
    \includegraphics[width=\linewidth,height=0.85\textheight,keepaspectratio]{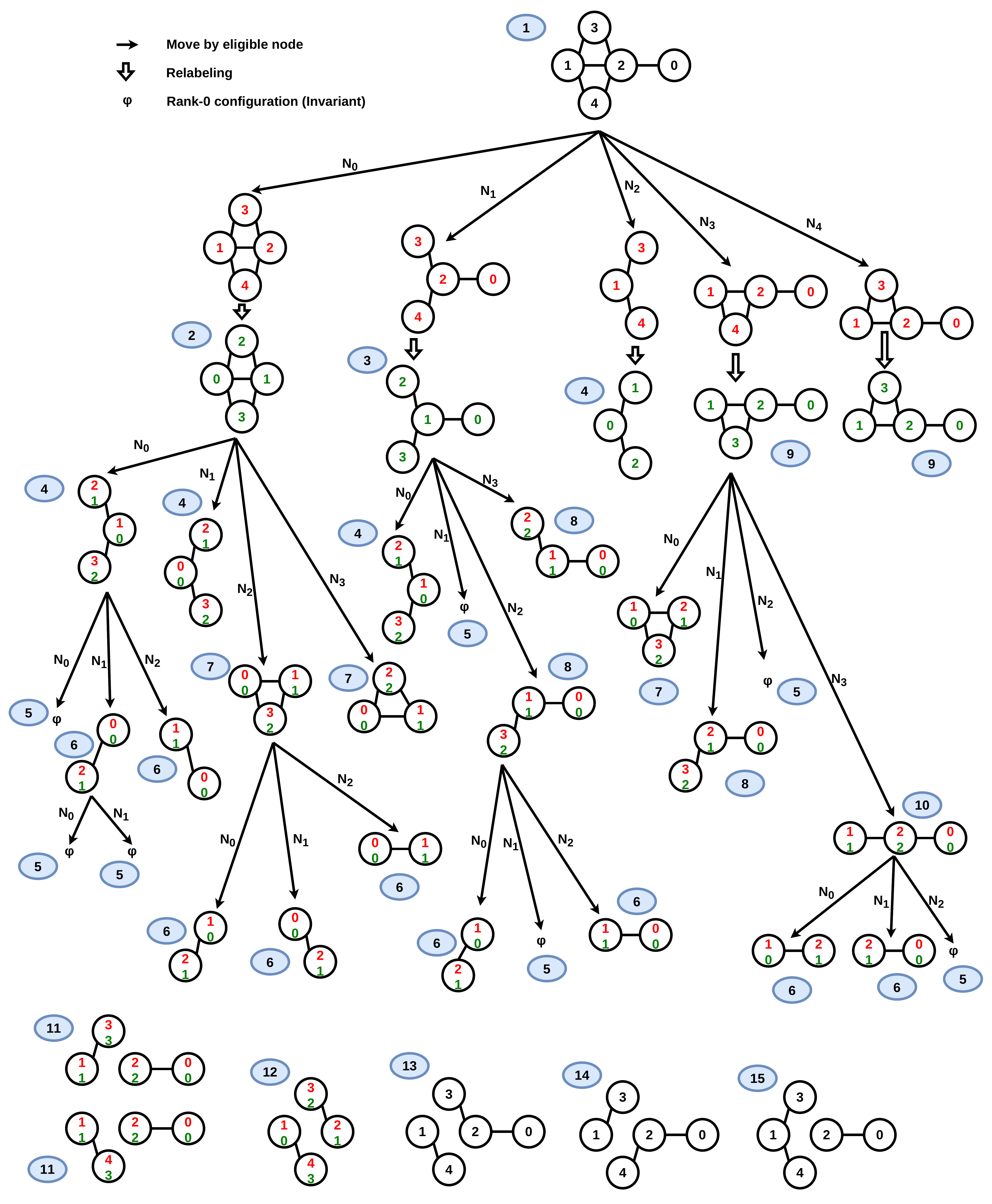}
    \captionof{figure}{\small Execution of the CE algorithm on a 5-node power-law cluster graph, where same-colored edges denote conflicts between connected nodes. Each layer branches on eligible node moves, removing the moved node and its incident edges, then relabels serially to the nearest available id (red: post-branching; green: post-relabeling). If relabeling produces a subgraph seen before, its rank is retrieved from cache instead of recomputed. Only $15$ cached subgraphs were needed to cover all $2 \times 4 \times 5 \times 3 \times 3 = 360$ states. Numbers in blue circles are unique subgraph ids explored.}
    \label{fig:ce-alg-execution}
\end{strip}

%% file: amit.bib
@inproceedings{NguyenGuptaKulkarni23ICDCNAnalyzing,
  title={Analyzing Program Transitions to Compute Benefit of Tolerating Consistency Violation Faults},
  author={Nguyen, Duong Ngoc and Gupta, Arya Tanmay and Kulkarni, Sandeep},
  booktitle={Proceedings of the 24th International Conference on Distributed Computing and Networking},
  pages={58--69},
  year={2023}
}

@article{Dijkstra1974CACM,
  title={Self-stabilizing systems in spite of distributed control},
  author={Dijkstra, Edsger W},
  journal={Communications of the ACM},
  volume={17},
  number={11},
  pages={643--644},
  year={1974},
  publisher={ACM New York, NY, USA}
}

@inproceedings{Nguyen2019ICDCN,
    author = {Nguyen, Duong and Kulkarni, Sandeep S. and Datta, Ajoy K.},
    title = {Benefit of self-stabilizing protocols in eventually consistent key-value stores: a case study},
    year = {2019},
    isbn = {9781450360944},
    publisher = {Association for Computing Machinery},
    address = {New York, NY, USA},
    url = {https://doi.org/10.1145/3288599.3288609},
    doi = {10.1145/3288599.3288609},
    booktitle = {Proceedings of the 20th International Conference on Distributed Computing and Networking},
    pages = {148–157},
    numpages = {10},
    location = {Bangalore, India},
    series = {ICDCN '19}
}

@misc{kipf2017semisupervisedclassificationgraphconvolutional,
      title={Semi-Supervised Classification with Graph Convolutional Networks}, 
      author={Thomas N. Kipf and Max Welling},
      year={2017},
      eprint={1609.02907},
      archivePrefix={arXiv},
      primaryClass={cs.LG},
      url={https://arxiv.org/abs/1609.02907}, 
}

@misc{hamilton2018inductiverepresentationlearninglarge,
      title={Inductive Representation Learning on Large Graphs}, 
      author={William L. Hamilton and Rex Ying and Jure Leskovec},
      year={2018},
      eprint={1706.02216},
      archivePrefix={arXiv},
      primaryClass={cs.SI},
      url={https://arxiv.org/abs/1706.02216}, 
}

@misc{veličković2018graphattentionnetworks,
    title={Graph Attention Networks}, 
    author={Petar Veličković and Guillem Cucurull and Arantxa Casanova and Adriana Romero and Pietro Liò and Yoshua Bengio},
    year={2018},
    eprint={1710.10903},
    archivePrefix={arXiv},
    primaryClass={stat.ML},
    url={https://arxiv.org/abs/1710.10903}, 
}

@inproceedings{DeCandiaHastorun07,
author = {DeCandia, Giuseppe and Hastorun, Deniz and Jampani, Madan and Kakulapati, Gunavardhan and Lakshman, Avinash and Pilchin, Alex and Sivasubramanian, Swaminathan and Vosshall, Peter and Vogels, Werner},
title = {Dynamo: amazon's highly available key-value store},
year = {2007},
isbn = {9781595935915},
publisher = {Association for Computing Machinery},
address = {New York, NY, USA},
url = {https://doi.org/10.1145/1294261.1294281},
doi = {10.1145/1294261.1294281},
booktitle = {Proceedings of Twenty-First ACM SIGOPS Symposium on Operating Systems Principles},
pages = {205–220},
numpages = {16},
location = {Stevenson, Washington, USA},
series = {SOSP '07}
}

@article{GilbertLynch02Brewer,
author = {Gilbert, Seth and Lynch, Nancy},
title = {Brewer's conjecture and the feasibility of consistent, available, partition-tolerant web services},
year = {2002},
issue_date = {June 2002},
publisher = {Association for Computing Machinery},
address = {New York, NY, USA},
volume = {33},
number = {2},
issn = {0163-5700},
url = {https://doi.org/10.1145/564585.564601},
doi = {10.1145/564585.564601},
journal = {SIGACT News},
month = jun,
pages = {51–59},
numpages = {9}
}

@article{spanner2013,
  title={Spanner: Google’s globally distributed database},
  author={Corbett, James C and Dean, Jeffrey and Epstein, Michael and Fikes, Andrew and Frost, Christopher and ...},
  journal={ACM Transactions on Computer Systems},
  volume={31},
  number={3},
  pages={8},
  year={2013},
  publisher={ACM}
}

@inproceedings{cassandra2009,
  title={Cassandra: A decentralized structured storage system},
  author={Lakshman, Avinash and Malik, Prashant},
  booktitle={Proceedings of the 2009 ACM SIGOPS/SOSP Workshop on Large‑Scale Distributed Systems and Middleware},
  year={2009}
}

@inproceedings{MalewiczMatthew10ICMD,
title	= {Pregel: a system for large-scale graph processing},author	= {Grzegorz Malewicz and Matthew H. Austern and Aart J.C Bik and James C. Dehnert and Ilan Horn and Naty Leiser and Grzegorz Czajkowski},year	= {2010},URL	= {http://doi.acm.org/10.1145/1807167.1807184},booktitle	= {Proceedings of the 2010 international conference on Management of data},pages	= {135--146},address	= {New York, NY, USA}}

@article{HEDETNIEMI2003251,
title = {Linear time self-stabilizing colorings},
journal = {Information Processing Letters},
volume = {87},
number = {5},
pages = {251-255},
year = {2003},
issn = {0020-0190},
doi = {https://doi.org/10.1016/S0020-0190(03)00299-0},
url = {https://www.sciencedirect.com/science/article/pii/S0020019003002990},
author = {Stephen T. Hedetniemi and David P. Jacobs and Pradip K. Srimani}
}

@article{TURAU200788,
title = {Linear self-stabilizing algorithms for the independent and dominating set problems using an unfair distributed scheduler},
journal = {Information Processing Letters},
volume = {103},
number = {3},
pages = {88-93},
year = {2007},
issn = {0020-0190},
doi = {https://doi.org/10.1016/j.ipl.2007.02.013},
url = {https://www.sciencedirect.com/science/article/pii/S0020019007000488},
author = {Volker Turau}
}


%% file: paper.bib
@article{faghih2018automated,
  title={Automated synthesis of distributed self-stabilizing protocols},
  author={Faghih, Fathiyeh and Bonakdarpour, Borzoo and Tixeuil, S{\'e}bastien and Kulkarni, Sandeep},
  journal={Logical Methods in Computer Science},
  volume={14},
  year={2018},
  publisher={Episciences. org}
}

@InProceedings{GosainSharma15StaticAnalysis,
author="Gosain, Anjana
and Sharma, Ganga",
editor="Mandal, Durbadal
and Kar, Rajib
and Das, Swagatam
and Panigrahi, Bijaya Ketan",
title="Static Analysis: A Survey of Techniques and Tools",
booktitle="Intelligent Computing and Applications",
year="2015",
publisher="Springer India",
address="New Delhi",
pages="581--591",
isbn="978-81-322-2268-2"
}

@article{zenil2014correlation,
  title={Correlation of automorphism group size and topological properties with program-size complexity evaluations of graphs and complex networks},
  author={Zenil, Hector and Soler-Toscano, Fernando and Dingle, Kamaludin and Louis, Ard A.},
  journal={Physica A: Statistical Mechanics and its Applications},
  volume={404},
  pages={341--358},
  year={2014},
  publisher={Elsevier}
}

@inproceedings{
    xu2018how,
    title={How Powerful are Graph Neural Networks?},
    author={Keyulu Xu and Weihua Hu and Jure Leskovec and Stefanie Jegelka},
    booktitle={International Conference on Learning Representations},
    year={2019},
    url={https://openreview.net/forum?id=ryGs6iA5Km},
}

@article{weisfeiler1968reduction,
  author  = {Weisfeiler, B. and Lehman, A. A.},
  title   = {A reduction of a graph to a canonical form and an algebra arising during this reduction},
  journal = {Nauchno-Technicheskaya Informatsia},
  series  = {2},
  number  = {9},
  year    = {1968}
}

@misc{pandurangan2018distributedcomplexitylargescalegraph,
      title={On the Distributed Complexity of Large-Scale Graph Computations}, 
      author={Gopal Pandurangan and Peter Robinson and Michele Scquizzato},
      year={2018},
      eprint={1602.08481},
      archivePrefix={arXiv},
      primaryClass={cs.DC},
      url={https://arxiv.org/abs/1602.08481}, 
}

@article{LiHudak89TOCS,
  author  = {Li, Kai and Hudak, Paul},
  title   = {Memory Coherence in Shared Virtual Memory Systems},
  journal = {ACM Transactions on Computer Systems},
  volume  = {7},
  number  = {4},
  pages   = {321--359},
  year    = {1989},
  doi     = {10.1145/75104.75105}
}

@incollection{valmari1998state,
  title     = {The State Explosion Problem},
  author    = {Valmari, Antti},
  booktitle = {Lectures on Petri Nets I: Basic Models},
  series    = {Lecture Notes in Computer Science},
  volume    = {1491},
  pages     = {429--528},
  publisher = {Springer},
  address   = {Berlin, Heidelberg},
  year      = {1998},
  doi       = {10.1007/3-540-65306-6_21}
}

@incollection{clarke2012model,
  title     = {Model Checking and the State Explosion Problem},
  author    = {Clarke, Edmund M. and Klieber, William and Nov{\'a}{\v{c}}ek, Milo{\v{s}} and Zuliani, Paolo},
  booktitle = {Tools for Practical Software Verification},
  series    = {Lecture Notes in Computer Science},
  volume    = {7682},
  publisher = {Springer},
  address   = {Berlin, Heidelberg},
  year      = {2012},
  doi       = {10.1007/978-3-642-35746-6_1}
}

@incollection{pelanek2009fighting,
  title     = {Fighting State Space Explosion: Review and Evaluation},
  author    = {Pel{\'a}nek, Radek},
  booktitle = {Formal Methods for Industrial Critical Systems (FMICS 2008)},
  series    = {Lecture Notes in Computer Science},
  volume    = {5596},
  pages     = {37--52},
  publisher = {Springer},
  address   = {Berlin, Heidelberg},
  year      = {2009},
  doi       = {10.1007/978-3-642-03240-0_7}
}

@misc{li2021oodgnnoutofdistributiongeneralizedgraph,
      title={OOD-GNN: Out-of-Distribution Generalized Graph Neural Network}, 
      author={Haoyang Li and Xin Wang and Ziwei Zhang and Wenwu Zhu},
      year={2021},
      eprint={2112.03806},
      archivePrefix={arXiv},
      primaryClass={cs.LG},
      url={https://arxiv.org/abs/2112.03806}, 
}

@inproceedings{nerem2025gnn,
  title     = {Graph Neural Networks Extrapolate Out-of-distribution for Shortest Paths},
  author    = {Nerem, Robert R. and others},
  booktitle = {Proceedings of Machine Learning Research},
  volume    = {336},
  year      = {2026},
  url       = {https://openreview.net/forum?id=zAdC99LPZU}
}

@article{chen2022causally,
  title   = {Learning Causally Invariant Representations for Out-of-Distribution Generalization on Graphs},
  author  = {Chen, Yongqiang and Zhang, Yonggang and Bian, Yatao and Yang, Han and Ma, Kaili and Xie, Binghui and Liu, Tongliang and Han, Bo and Cheng, James},
  journal = {Advances in Neural Information Processing Systems},
  volume  = {35},
  year    = {2022}
}

@inproceedings{wu2024causal,
  title     = {Graph Out-of-Distribution Generalization via Causal Intervention},
  author    = {Wu, Qitian and Nie, Fan and Yang, Chenxiao and Bao, Tianyi and Yan, Junchi},
  booktitle = {Proceedings of the ACM Web Conference (WWW)},
  year      = {2024}
}

@article{li2023graphstructure,
  title   = {Graph Structure and Feature Extrapolation for Out-of-Distribution Generalization},
  author  = {Li, Xiner and Gui, Shurui and Luo, Youzhi and Ji, Shuiwang},
  journal = {arXiv preprint arXiv:2306.08076},
  year    = {2023}
}

@inproceedings{buffelli2022sizeshiftreg,
  title     = {SizeShiftReg: A Regularization Method for Improving Size-Generalization in Graph Neural Networks},
  author    = {Buffelli, Davide and Li{\`o}, Pietro and Vandin, Fabio},
  booktitle = {Advances in Neural Information Processing Systems (NeurIPS)},
  year      = {2022}
}

@inproceedings{yamauchi2010monotonic,
  title     = {Monotonic Stabilization},
  author    = {Yamauchi, Yukiko and Tixeuil, S{\'e}bastien},
  booktitle = {Principles of Distributed Systems (OPODIS 2010)},
  series    = {Lecture Notes in Computer Science},
  volume    = {6490},
  pages     = {475--490},
  publisher = {Springer},
  address   = {Berlin, Heidelberg},
  year      = {2010},
  doi       = {10.1007/978-3-642-17653-1_34}
}

@article{tsuchiya2001symbolic,
  title   = {Symbolic Model Checking for Self-Stabilizing Algorithms},
  author  = {Tsuchiya, Tatsuhiro and Nagano, Shin and Paidi, Rajit B. and Kikuno, Tohru},
  journal = {IEEE Transactions on Parallel and Distributed Systems},
  volume  = {12},
  number  = {1},
  pages   = {81--95},
  year    = {2001},
  doi     = {10.1109/71.899941}
}

@article{jaubert2012scalable,
  title   = {Towards Scalable Model Checking of Self-Stabilizing Programs},
  journal = {Journal of Parallel and Distributed Computing},
  year    = {2012},
  doi     = {10.1016/j.jpdc.2012.11.008}
}

@article{cai1992optimal,
  author  = {Cai, Jin-Yi and F{\"u}rer, Martin and Immerman, Neil},
  title   = {An Optimal Lower Bound on the Number of Variables for
             Graph Identification},
  journal = {Combinatorica},
  volume  = {12},
  number  = {4},
  pages   = {389--410},
  year    = {1992}
}

@article{mckay2014practical,
  author  = {McKay, Brendan D. and Piperno, Adolfo},
  title   = {Practical Graph Isomorphism, {II}},
  journal = {Journal of Symbolic Computation},
  volume  = {60},
  pages   = {94--112},
  year    = {2014}
}
